\documentclass[a4paper,11pt]{article}

\usepackage{color}
\usepackage{amssymb,amsmath,amsthm,verbatim}
\usepackage{url}

\def\n{\noindent}
\def\pp{{\mathbf p}}
\def\cc{{\mathbf c}}

\def\mt{t\kern-0.035cm\char39\kern-0.03cm}
\def\ml{l\kern-0.035cm\char39\kern-0.03cm}
\def\md{d\kern-0.035cm\char39\kern-0.03cm}
\newtheorem{thm}{Theorem}

\newtheorem{lema}{Lemma}
\newtheorem{df}{Definition}
\newtheorem{cor}{Corollary}

\newenvironment{ex}{{\bf Example}}{\bigskip}

\def\MM{{\cal M}}

\def\sat{\small (2,2)\normalsize-{\sc e}\small 3\normalsize-{\sc sat}\ }
\def\satn{\small (2,2)\normalsize-{\sc e}\small 3\normalsize-{\sc sat}}
\begin{document}

\title{\bf Stable matchings of teachers to schools\thanks{This work was supported by VEGA grants 1/0344/14 and 1/0142/15 (Cechl\'arov\'a), by OTKA grant K108383 and the ELTE-MTA Egerv\'ary Research Group (Fleiner), by EPSRC grant EP/K010042/1 (Manlove) and by a SICSA Prize Studentship (McBride). The authors also gratefully acknowledge the support of COST Action IC1205 on Computational Social Choice.}}
\author{Katar\'\i na Cechl\'arov\'a$^1$,  Tam\'as Fleiner$^2$, David F.\ Manlove$^3$ and Iain McBride$^3$\\ \\
				\footnotesize{$^1$Institute of Mathematics,
				Faculty of Science, P.J. \v Saf\'arik University, }\\
			\footnotesize{Jesenn\'a 5, 040 01 Ko\v sice, Slovakia. Email {\tt katarina.cechlarova@upjs.sk}}\\
				 \footnotesize{$^2$Department of Computer Science and Information Theory,}\\ 
			 \footnotesize{Budapest University of Technology and Economics, Magyar tud\'osok k\"or\'utja 2, H-1117 Budapest,}\\
 	\footnotesize{Hungary and MTA-ELTE Egerv\'ary Research Group.  Email {\tt fleiner@cs.bme.hu}}\\
			\footnotesize{$^3$School of Computing Science,  Sir Alwyn Williams Building, University of Glasgow,}\\
\footnotesize{Glasgow, G12 8QQ, UK.  Email {\tt david.manlove@glasgow.ac.uk}, {\tt i.mcbride.1@research.gla.ac.uk}}
}
\date{ }
\maketitle
\begin{abstract} Several countries successfully use centralized matching schemes for school or higher education assignment, or for entry-level labour markets.  In this paper we explore the computational aspects of a possible similar scheme for assigning teachers to schools. Our model is motivated by a particular characteristic of the education system in many countries where each  teacher specializes in two subjects.  We seek stable matchings, which ensure that no teacher and school have the incentive to deviate from their assignments.  Indeed we propose two stability definitions depending on the precise format of schools' preferences.  If the schools' ranking of applicants is independent of their subjects of specialism, we show that the problem of deciding whether a stable matching exists is NP-complete, even if there are only three subjects, unless there are master lists of applicants or of schools. By contrast, if the schools may order applicants differently in each of their specialization subjects, the problem of deciding whether a stable matching exists is NP-complete even in the presence of subject-specific master lists plus a master list of schools. Finally, we prove a strong inapproximability result for the problem of finding a matching with the minimum number of blocking pairs with respect to both stability definitions.
\end{abstract}

\\

\section{Introduction}

In the organization of education, several countries or regions use various centralized schemes to allocate children to public schools (e.g., in Boston and New York \cite{APR05,APR05a}),  students to universities (e.g., in Hungary \cite{BFIM10}), and intending junior doctors to training positions in hospitals (e.g., in the USA \cite{ZZZ5}), etc. These schemes are usually not dictatorial in the sense that they take into account the wishes of both sides of the market: students  may express their preferences over the universities they wish to attend, and the universities may order their applicants based on some kind of evaluation. After analyzing several successful and unsuccessful schemes Roth \cite{Rot84, Rot08} convincingly argued that a crucial property for success is so-called {\it stability}, introduced in the seminal paper by Gale and Shapley \cite{GS62}. Stability means that no unmatched student--school pair should simultaneously prefer each other to their current assignee(s) (if any). In many real markets, each  instance not only admits a stable matching, but it is also possible to find such a matching efficiently.

However, sometimes there are circumstances  leading to additional structural requirements. For example, married couples may wish to be allocated to the same hospital or at least to hospitals that are geographically close \cite{MM10}, or schools may wish to have the right to close a study programme if the number of applicants does not meet a certain lower quota \cite{BFIM10}. In such cases, a suitable notion of stability has to be defined that really mirrors the intentions of the participants and motivates them to obey the recommended assignment. Alas, a stable matching is not necessarily bound to exist; and even worse, it is often a computationally difficult problem to decide whether in the given situation one does exist \cite{Ron90}.

The topic of this paper is motivated by the problems arising in  the labour market for teachers. Traditionally, a teacher for the upper elementary or lower secondary level of education in  Slovakia and the Czech republic (and in fact in many other countries and regions, such as Germany \cite{BAUM10} and Flanders \cite{VDVO03}) specializes in two curricular domains (from now on called {\it subjects}), e.g., Mathematics and Physics, Chemistry and Biology, or Slovak language and English etc. 
When a school is looking for new teachers, it may have different  numbers of lessons (or teaching hours) to cover at each subject. Thus we suppose that each school has different capacity for each subject and that it will be willing to employ a set of teachers in such a way that these capacities will not be exceeded.  Cechl\'arov\'a et al. \cite{CFMMP13} studied a variant of this problem where the tranee-teachers could only express which schools are acceptable for them and which not without ordering them according to their preferences and the schools had no say. In these settings, the aim was to assign as many trainee-teachers as possible (ideally all of them) by respecting the schools capacities.

The aim of this paper is to study algorithmic aspects of the problem of assigning teachers to schools within the framework of two-sided preferences. We suppose that teachers rank order their acceptable schools according to their own criteria, and vice versa, schools rank order their applicants similarly \cite{Man13}.  In this context we suggest two stability  definitions and study the computational complexity of problems concerned with finding stable matchings (or reporting that none exist).  These definitions and the associated complexity results depend on the nature of the schools' preference lists.

The main results and the organization of the rest of the paper are as follows. In Section \ref{s_defi} we introduce relevant technical concepts and illustrate them by means of simple examples. In Section \ref{s_linear} we deal with the case when each school has a linear ordering on the set of teachers who apply for a position. We show that in this general case the problem of deciding whether a stable matching exists is NP-complete, even if there are only three subjects in total. By contrast, if either the preferences of schools are derived from a common master list of teachers, or vice versa,  if the preferences of teachers are derived from a common master list of schools, a unique stable matching exists and it can be found using straightforward extensions of the classical serial Dictatorship mechanism \cite{RS90}.  Moreover, the problems with master lists are efficiently solvable without any restrictions on the number of subjects. In Section \ref{s_subject_specific}  we modify the stability definition to enable the schools to order the teachers differently according to their two specialization subjects. We show that in this case, the problem of deciding whether a stable matching exists is NP-complete even if there are only three subjects, there are master lists for each subject and there is also a master list of schools.  Problems involving finding matchings with the minimum number of blocking pairs are discussed in Section \ref{sec:moststable}.  Finally Section \ref{s_conclusion} summarizes our findings and suggests some questions for further research. 

\section{Preliminary definitions and observations}\label{s_defi}
An instance of the Teachers Assignment Problem, {\sc tap} for short, involves a set $A$ of applicants (teachers), a set $S$ of schools and a set $P$ of subjects. 
For ease of exposition, elements of the set $P$ will sometimes be referred to by letters like $M$, $F$ or  $I$ 
to remind the reader of real subjects taught at schools, such as Mathematics, Physics, or Informatics, etc.
 
Each applicant $a\in A$ is characterized by a pair of distinct subjects $\pp(a)\subseteq P$, where $\pp(a)=\{p_1(a),p_2(a)\}$, that define her {\em type}. Sometimes we shall also say that a particular applicant is of type $FM$, $IM$, or $FI$, etc.  Corresponding to each applicant $a\in A$ there is a set $S(a)\subseteq S$ of schools that $a$ finds \emph{acceptable}.  Moreover applicant $a$ ranks $S(a)$ in strict order of preference.

Each school $s\in S$ has a certain capacity for each subject: the vector of capacities will be denoted by $\cc(s)=(c_1(s),\dots, c_k(s))\in {\mathbb N}^k$, where $k=|P|$, and an entry of $\cc(s)$ will be called the {\em partial capacity} of school $s$. Here, $c_i(s)$ is the maximum number of applicants, whose specialization involves subject $p_i$, that school $s$ is able to take. 

Let $S(A)=\{(a,s) : a\in A\wedge s\in S(a)\}$ denote the set of \emph{acceptable} applicant--school pairs.  An {\em assignment} $\MM$ is a subset of $S(A)$ such that each applicant $a\in A$ is a member of at most one pair in $\MM$. We shall write $\MM(a)=s$ if $(a,s)\in \MM$ and say that applicant $a$ is assigned to school $s$ and $\MM(a)=\emptyset$ if there is no $s\in S$ with $(a,s)\in \MM$. The set of applicants assigned to a school $s$ will be denoted by $\MM(s)=\{a\in A : (a,s)\in \MM\}$. We shall also denote by $\MM_p(s)$ the set of applicants assigned to $s$ whose specialization includes subject $p$ and by $\MM_{p,r}(s)$ the set of applicants assigned to $s$ whose specialization is exactly the pair $\{p,r\}$.  More precisely,
$$\MM_p(s)=\{a\in A : (a,s)\in \MM \ \wedge\ p\in \pp(a)\}$$ and $$\MM_{p,r}(s)=\{a\in A : (a,s)\in \MM \  \wedge\  \{p,r\}= \pp(a)\}.$$
An assignment $\MM$ is a {\em matching} if $|\MM_p(s)|\leq c_p(s)$ for each school $s$ and each subject $p$.  We say that an applicant $a$ is {\em assigned} in $\MM$ if $\MM(a)\ne\emptyset$, otherwise she is {\em unassigned}. A school $s$ is {\em full} in a matching $\MM$ if it can admit no other student (irrespective of her specialization) and that it is {\em undersubscribed} in subject $p$ if $|\MM_p(s)|< c_p(s)$.

\begin{df}\label{d_blok} Let $\MM$ be a matching. We say  that a pair $(a,s)$  with $\pp(a)=\{p_1,p_2\}$ and $s\in S(a)$ is \emph{blocking} if $a$ is not assigned in $\MM$ or $a$ prefers $s$  to $\MM(a)$ and one of the following conditions hold:
\begin{itemize}\itemsep0pt
\item[$(i)$] $s$ is undersubscribed in both $p_1$ and $p_2$, 
\item[$(ii)$] $s$ is undersubscribed in $p_i$ and it prefers $a$ to one applicant in $\MM_{p_{3-i}}(s)$ for some $i\in\{1,2\}$, 
\item[$(iii)$] $s$ prefers $a$ to one applicant in $\MM_{p_1,p_2}(s)$,
\item[$(iv)$] $s$ prefers $a$ to two different applicants $a_1,a_2$ such that $a_1\in \MM_{p_1}(s)$ and $a_2\in \MM_{p_2}(s)$.
\end{itemize}
A matching is \emph{stable} if it  admits no blocking pair. 
\end{df}

\n
\begin{ex} {\bf 1.} Let $J_1$ be the instance of {\sc tap} with the set of subjects $P=\{F,I,M\}$ given in Figure \ref{f_1}.
\begin{figure}[ht]
\begin{displaymath}
\begin{array}{cllcll}
\mbox{applicant} &\mbox{type}\ &\mbox{preferences}\qquad \qquad & 
\mbox{school} & \mbox{capacities} &\mbox{preferences}\\
&&&& F\quad I\quad M \\ 
\noalign{\hrule} 
\\[-8pt]
a_1 & MF & s_1,s_3  & s_1 & 2\quad\ 1\quad\ 1 & a_3,a_4,a_1,a_2 \\
a_2 & MF & s_1,s_3  & s_2 & 1\quad\ 1\quad\ 1 & a_4,a_3         \\
a_3 & MI & s_1,s_2  & s_3 & 2\quad\ 1\quad\ 1 & a_4,a_1,a_2     \\
a_4 & IF & s_3,s_2,s_1   &                   &                  \\
\noalign{\hrule}
\end{array}
\end{displaymath}
\caption{Instance $J_1$ of {\sc tap}}\label{f_1}
\end{figure} 

\n Consider the matching $\MM=\{(a_1,s_3),(a_2,s_1),(a_3,s_2),(a_4,s_1)\}$. 
It is easy to see that $\MM$ is not stable. Each of the conditions $(i)-(iv)$ of Definition \ref{d_blok} is violated by the following blocking pairs, respectively:   
\begin{itemize}\itemsep0pt
\item[$(i)$] $(a_4,s_3)$ is a blocking pair since $a_4$ prefers school $s_3$ to $\MM(a_4)=s_1$, and $s_3$ is undersubscribed in both $I$ and $F$, 
\item[$(ii)$] $(a_4,s_2)$ is a blocking pair since $a_4$ prefers $s_2$ to $\MM(a_4)=s_1$, school $s_2$ is undersubscribed in $F$ and it prefers $a_4$ to $a_3\in \MM(s_2)$, 
\item[$(iii)$] $(a_1,s_1)$ is a blocking pair since $a_1$ prefers $s_1$ to $\MM(a_1)=s_3$, and school $s_1$ prefers $a_3$ to $a_2\in \MM(s_1)$ who is of the same type as $a_1$,
\item[$(iv)$] $(a_3,s_1)$ is a blocking pair since $a_3$ prefers $s_1$ to $\MM(a_1)=s_2$, and school $s_1$ prefers $a_2$ to both its asignees $a_2$ and $a_4$.
\end{itemize}
\label{ex_1}
\end{ex}

Note that Definition \ref{d_blok} $(iv)$, as also illustrated  in Example 1, 
gives rise to the possibility that a school could drop two applicants and accept just one in order to satisfy a blocking pair.  This is also a possibility in the Hospitals / Residents problem with Couples where a single resident $r$ can displace a couple $c$ assigned to a hospital $h$ if $h$ prefers $r$ to just one member of $c$ \cite{BIS11}.

We also remark that {\sc tap} bears a superficial resemblance to the variant of the Hospitals / Residents problem that modelled the problem of assigning junior doctors to hospitals in Scotland in years 2000-2005, where intending junior doctors sought not one position at hospitals, but two, namely a medical post and a surgical post \cite{Irv08}.  They also typically had preferences over the half-years in which they would carry out each type of post, so the stability definition was different to the one given in Definition \ref{d_blok}.

We next present two examples to show that a {\sc tap} instance need not admit a stable matching, and in such instances that do, stable matchings may have different sizes.
\bigskip

\noindent \begin{ex} {\bf 2.} Consider the instance $J_2$ of {\sc tap} given in Figure \ref{f_2}. We show that $J_2$ admits no stable matching. 
If $\MM(a_1)=s_2$ then assigning $a_2$ to $s_1$ leads to the blocking pair $(a_3,s_1)$ and
assigning $a_3$ to $s_1$ produces blocking pair $(a_2,s_2)$.
By contrast, if $\MM(a_1)=s_1$ then $a_2$ must also be assigned to $s_1$ (this schools is her first choice and it has enough place to accept her), which makes the pair $(a_1,s_2)$ blocking. 
Hence no stable matching exists.
\begin{figure}[ht]
\begin{displaymath}
\begin{array}{rllrll}
\mbox{applicant} &\mbox{type}\ &\mbox{preferences}\qquad \qquad & 
\mbox{school} & \mbox{capacities} &\mbox{preferences}\\
&&&& F\quad I\quad M \\ 
\noalign{\hrule} 
\\[-8pt]
a_1 & FM & s_2,s_1  & s_1 & 1\quad\ 1\quad\ 2 & a_1,a_3,a_2 \\
a_2 & IM & s_1,s_2  & s_2 & 1\quad\ 1\quad\ 1 & a_2, a_1    \\
a_3 & FI & s_1      &     &  \quad\  \quad\   &             \\
\noalign{\hrule}
\end{array}
\end{displaymath}
\caption{An instance $J_2$ of {\sc tap} with no stable matching}\label{f_2}
\end{figure} 
\end{ex}

\bigskip
\noindent \begin{ex} {\bf 3.}  Consider the instance $J_3$ of {\sc tap} given in Figure \ref{f_3}.  It is straightforward to verify that $\MM_1=\{(a_1,s_2),(a_2,s_1)\}$, of size 2, and $\MM_2=\{(a_1,s_1),(a_2,s_2),(a_3,s_1)\}$, of size 3, are both stable in $J_3$.  Hence $J_3$ admits stable matchings of different sizes.
\begin{figure}[ht]
\begin{displaymath}
\begin{array}{rllrll}
\mbox{applicant} &\mbox{type}\ &\mbox{preferences}\qquad \qquad & 
\mbox{school} & \mbox{capacities} &\mbox{preferences}\\
&&&& F\quad I\quad M \\ 
\noalign{\hrule} 
\\[-8pt]
a_1 & FM & s_2,s_1  & s_1 & 1\quad\ 1\quad\ 2 & a_1,a_2,a_3 \\
a_2 & FI & s_1,s_2  & s_2 & 1\quad\ 1\quad\ 1 & a_2, a_1    \\
a_3 & IM & s_1      &     &  \quad\  \quad\   &             \\
\noalign{\hrule}
\end{array}
\end{displaymath}
\caption{An instance $J_3$ of {\sc tap} with stable matchings of different sizes}\label{f_3}
\end{figure} 
\end{ex}

\section{NP-hardness of TAP}\label{s_linear}
\begin{thm}\label{thm1}
Given an instance of {\sc tap}, the problem of deciding whether a stable matching exists, is NP-complete. This result holds even if there are at most three subjects, each partial capacity of a school is at most 2, and the preference list of each applicant is of length at most 3.
\end{thm}
\begin{proof} It is easy to  see that {\sc tap} belongs to NP, since when given an assignment, it can be checked in polynomial time that it is a matching and that it is stable. To prove completeness,  we reduce from a  restricted version of {\sc sat}. Let \sat denote the problem of deciding, given a Boolean formula $B$ in CNF in which each clause contains exactly 3 literals and, for each variable $v_i$, each of literals $v_i$ and $\bar{v}_i$ appears exactly twice in $B$, whether $B$ is satisfiable. Berman et al. \cite{BKS03} showed that \sat is NP-complete.

Hence let $B$ be an instance of \satn. Let $V = \{v_1, v_2,\dots, v_n\}$ and $C =\{c_1, c_2,\dots, c_m\}$ be the set of variables and clauses  in $B$, respectively. Let us construct an instance $J$ of {\sc tap} in the following way.

There are 3 subjects, namely $F$, $I$ and $M$. For each variable $v_i$ there are 6 applicants $a_i^1, a_i^2,\dots, a_i^6$, 4 applicants $x_i^1,x_i^2,y_i^1,y_i^2$, 12 applicants $q_{i,1}^k,q_{i,2}^k,q_{i,3}^k$ ($1\leq k\leq 4$), 6 schools $s_i^1,s_i^2,s_i^3,s_i^4, s_i^T,s_i^F$ and 12 schools $w_{i,1}^k,w_{i,2}^k,w_{i,3}^k$ ($1\leq k\leq 4$). For each clause $c_j$ there is one school $z_j$. Applicants $x_i^1$ and $x_i^2$ correspond to the first and to the second occurrence of literal $v_i$, and applicants $y_i^1$ and $y_i^2$ correspond to the first and to the second occurrence of literal $\bar{v}_i$, respectively.    

The characteristics of applicants and schools and their preferences are given in Figure \ref{t_1}. 
\begin{figure}[ht]
\begin{displaymath}
\begin{array}{rllrll}
\mbox{applicant} &\mbox{type}\ &\mbox{preferences}\qquad \qquad & 
\mbox{school} & \mbox{capacities} &\mbox{preferences}\\
&&&& F\quad I\quad M \\ 
\noalign{\hrule} 
\\[-8pt]
a_i^1 & FI & s_i^1,s_i^3 & s_i^1 & 1\quad\ 1\quad\ 2 & a_i^4,a_i^1,a_i^3 \\[1.5mm]
a_i^2 & FI & s_i^2,s_i^4 & s_i^2 & 1\quad\ 1\quad\ 2 & a_i^3,a_i^2,a_i^4 \\[1.5mm]
a_i^3 & FM & s_i^1,s_i^2 & s_i^3 & 1\quad\ 1\quad\ 0 & a_i^1,a_i^5       \\[1.5mm]
a_i^4 & IM & s_i^2,s_i^1 & s_i^4 & 1\quad\ 1\quad\ 0 & a_i^2,a_i^6       \\[1.5mm]
a_i^5 & FI & s_i^3,s_i^T & s_i^T & 1\quad\ 1\quad\ 2 & a_i^5,x_i^1,x_i^2 \\[1.5mm]
a_i^6 & FI & s_i^4,s_i^F & s_i^F & 1\quad\ 1\quad\ 2 & a_i^6,y_i^1,y_i^2 \\[1.5mm]
x_i^1 & FM & s_i^T,c(x_i^1),w_{i,3}^1 & w_{i,1}^k & 1\quad\ 1\quad\ 2 & q_{i,1}^k,q_{i,3}^k,q_{i,2}^k \\[1.5mm]
x_i^2 & IM & s_i^T,c(x_i^2),w_{i,3}^2 & w_{i,2}^k & 1\quad\ 1\quad\ 2 & q_{i,2}^k,q_{i,1}^k \\[1.5mm]
y_i^1 & FM & s_i^F,c(y_i^1),w_{i,3}^3 & w_{i,3}^k & 1\quad\ 1\quad\ 1 & A(w_{i,3}^k),q_{i,3}^k \\[1.5mm]
y_i^2 & IM & s_i^F,c(y_i^2),w_{i,3}^4 & z_j  & 2\quad\ 2\quad\ 2 & v_j^1,v_j^2,v_j^3 \\[1.5mm]
q_{i,1}^k & FM & w_{i,2}^k,w_{i,1}^k & & &  \\[1.5mm]
q_{i,2}^k & IM & w_{i,1}^k,w_{i,2}^k & & &  \\[1.5mm]
q_{i,3}^k & FI & w_{i,3}^k,w_{i,1}^k & & & \\[1.5mm]
\noalign{\hrule}
\end{array}
\end{displaymath}
\caption{The {\sc tap} instance constructed in the proof of Theorem \ref{thm1}.}\label{t_1}
\end{figure} 
Here, the subscripts and superscripts involving $i$, $j$ and $k$ range over the following intervals: $1\leq i\leq n$, $1\leq j\leq m$ and $1\leq k\leq 4$.  In the preference list of school $z_j$, the symbol $v_j^s$ means the $x$- or $y$-applicant that corresponds to the literal that appears in position $s$ of clause $c_j$. Conversely, in the preference list of $x$- or $y$-applicants the symbol $c(.)$ denotes the $z$-school corresponding to the clause containing the corresponding literal.  Also, in the preference list of $w_{i,3}^k$, the symbol $A(w_{i,3}^k)$ denotes $x_i^k$ if $1\leq k\leq 2$ and denotes $y_i^{k-2}$ if $3\leq k\leq 4$.

For each $i$ ($1\leq i\leq n$) let us denote
$$T_i=\{(x_i^1,s_i^T), (x_i^2,s_i^T), (a_i^6, s_i^F)\}, \qquad F_i=\{(y_i^1,s_i^F), (y_i^2,s_i^F), (a_i^5, s_i^T)\}.$$

Now, let $f$ be a satisfying truth assignment of $B$. Define a matching $\MM$ in $J$ as follows.
For each variable $v_i\in V$, if $v_i$ is true under $f$, put the pairs $T_i$ into  $\MM$ and  if $v_i$ is false under $f$ put the pairs $F_i$ into  $\MM$. In the former case add the pairs 
$$(y_i^1,c(y_i^1)), (y_i^2,c(y_i^2)), (a_i^1,s_i^1), (a_i^2,s_i^4), (a_i^3,s_i^2), (a_i^4,s_i^2), (a_i^5,s_i^3),$$ 
and in the latter case add the pairs
$$(x_i^1,c(x_i^1)), (x_i^2,c(x_i^2)), (a_i^1,s_i^3), (a_i^2,s_i^2), (a_i^3,s_i^1), (a_i^4,s_i^1), (a_i^6,s_i^4).$$ 
Notice that as each clause  $c_j\in C$ contains at most two false literals, school $z_j$ has enough capacity for accepting all the allocated applicants.  Finally, add the following pairs for each $i$ ($1\leq i\leq n$) and $k$ ($1\leq k\leq 4$):
$$(q_{i,1}^k,w_{i,2}^k), (q_{i,2}^k,w_{i,1}^k), (q_{i,3}^k,w_{i,3}^k).$$

It is obvious that the defined assignment is a matching; it remains to prove that it is stable. We show this by considering each type of applicants corresponding to variable $v_i$ in turn. Firstly we remark that applicants $q_{i,1}^k,q_{i,2}^k,q_{i,3}^k$ each have their first choice school $(1\leq k\leq 4)$ so cannot be involved in a blocking pair.  Now suppose that $v_i$ is true under $f$. Then:
\begin{itemize}
\item applicants $x_i^1$, $x_i^2$, $a_i^1$, $a_i^4$ and $a_i^5$ have their most-preferred schools, so are not blocking;
\item applicants $y_i^1$ and $y_i^2$  prefer school $s_i^F$, but this school is assigned $a_i^6$, whom it prefers;
\item applicant $a_i^2$ prefers school $s_i^2$, but this school is assigned $a_i^3$, whom it prefers;
\item applicant $a_i^3$ prefers school $s_i^1$, but this school is assigned $a_i^1$, whom it prefers;
\item applicant $a_i^6$ prefers school $s_i^4$, but this school is assigned $a_i^2$, whom it prefers.
\end{itemize}
The case of a false variable can be proved similarly.

For the converse implication let us first prove two lemmata.

\begin{lema}\label{l1} Each stable matching $\MM$ in $J$ contains for each $i$ either all the pairs in $T_i$ or all the pairs in $F_i$.
\end{lema}
\begin{proof}
Let $\MM$ be a stable matching. Fix $i\in \{1,2,\dots,n\}$. Notice first that both schools $s_i^T$ and $s_i^F$ must be full, otherwise either $s_i^T$  will form a blocking pair with at least  one of $x_i^1$ and $x_i^2$, or $s_i^F$  will form a blocking pair with at least  one of $y_i^1$ and $y_i^2$. Further, let us distinguish the following cases.
\begin{itemize}
\item $\{(a_i^5, s_i^T), (a_i^6, s_i^F)\}\subseteq \MM$. Then, as there are no blocking pairs,  $\{(a_i^1, s_i^3), (a_i^2, s_i^4)\}\subseteq \MM$, which further implies  $\{(a_i^3, s_i^2), (a_i^4, s_i^1)\}\subseteq \MM$. This, however means that  $(a_i^3, s_i^1)$ and $(a_i^4, s_i^2)$ are blocking pairs for $\MM$, a contradiction.
\item $\{(x_i^1, s_i^T),(x_i^2, s_i^T), (y_i^1, s_i^F), (y_i^2, s_i^F)\}\subseteq \MM$. Now, to avoid blocking pairs,   $\{(a_i^5, s_i^3),$ $(a_i^6, s_i^4)\}$ $\subseteq \MM$, which further implies  $\{(a_i^1, s_i^1), (a_i^2, s_i^2)\}\subseteq \MM$. Then there are blocking pairs  $(a_i^3, s_i^2)$ and $(a_i^4, s_i^1)$, again a contradiction.
\end{itemize}
The result follows.
\end{proof}

\begin{lema}\label{l2} In each stable matching $\MM$ in $J$, every applicant in the set $\{x_i^1,x_i^2,y_i^1,y_i^2 : 1\leq i\leq n\}$ is assigned to her first- or second-choice school.
\end{lema}
\begin{proof}
For some $i\in\{1,2,\dots,n\}$, consider applicant $x_i^1$ (the argument for $x_i^2$, $y_i^1$, $y_i^2$ is similar).  Suppose firstly that $x_i^1$ is unassigned in $\MM$.  Then $(x_i^1,w_{i,3}^1)$ blocks $\MM$, a contradiction.  Now suppose that $(x_i^1,w_{i,3}^1)\in \MM$.  If $(q_{i,3}^1,w_{i,1}^1)\in \MM$ then $(q_{i,1}^1,w_{i,2}^1)\in \MM$, for otherwise $(q_{i,1}^1,w_{i,1}^1)$ blocks $\MM$.  But then $(q_{i,2}^1,w_{i,2}^1)$ blocks $\MM$, a contradiction.  Thus $q_{i,3}^1$ is unassigned in $\MM$.  Then $(q_{i,2}^1,w_{i,1}^1)\in \MM$, for otherwise  $(q_{i,2}^1,w_{i,1}^1)$ blocks $\MM$.  Also $(q_{i,1}^1,w_{i,2}^1)\in \MM$, for otherwise $(q_{i,1}^1,w_{i,2}^1)$ blocks $\MM$.  Hence $(q_{i,3}^1,w_{i,1}^1)$ blocks $\MM$, a contradiction.
\end{proof}

So, suppose that $\MM$ is a stable matching in $J$.   We form a truth assignment $f$ in $B$ as follows. Let $i\in\{1,2,\dots,n\}$ be given.  By Lemma 1, either $T_i\subseteq \MM$  or $F_i\subseteq \MM$.  In the former case set $f(v_i)={\sf true}$, otherwise set $f(v_i)={\sf false}$.  Now let $v_i\in V$ and suppose that $f(v_i)={\sf true}$.  Then by Lemma 2, each of $y_{i,1}$ and $y_{i,2}$ is assigned to her second choice school.  Now suppose that $f(v_i)={\sf false}$.  Then again by Lemma 2, each of $x_{i,1}$ and $x_{i,2}$ is assigned to her second choice school.  Now let $c_j\in C$ and suppose that all literals in $c_j$ are false.  By the preceding remarks about $x_{i,1},x_{i,2},y_{i,1}$ and $y_{i,2}$ we deduce that $z_j$ is over-subscribed, a contradiction.  Thus $f$ is a satisfying truth assignment.
\end{proof}

\section{Master lists}\label{s_master}

In some centralized matching schemes all the applicants are ordered in a common {\it master} list. Although the criteria used for creating such lists are often subject to some controversy (see  \cite{Haw06} for the description of the matching scheme for allocating  medical students to hospital posts in England in 2005-06 and \cite{Tom05} for the situation in the central allocation scheme of teachers in Portugal that was used prior to 2005), computationally the situation with master lists may be easier.  A detailed study of stable matching problems with master lists from the computational point of view can be found in \cite{IMS08}.  In this section we shall consider the case of a master list of applicants and the (perhaps slightly less realistic) case of a master list of schools.

The problems of deciding the existence of a stable matching in these cases will be denoted by {\sc tap-am} and   {\sc tap-sm}, respectively. The phrase '$s$ has enough capacity for $a$' used in the algorithms in this section means the following: if $a$ is of type $\{p,r\}$ then $|\MM_p(s)| < c_p(s)$ as well as $|\MM_r(s)| < c_r(s)$. 

{\small
\begin{figure}
\noindent\rule{14cm}{0.5pt}
\begin{tabbing}
~~~ \= ~~~ \= ~~~ \= ~~~ \= \kill
{\bf begin}\\
\> $\MM:=\emptyset$;\\
\>  {\bf for\ } $i=1,2,\dots,n$ \\
\>  \>  {\bf if\ } $a_i$'s list contains a school with enough free capacity for $a_i$\{ \\
\>  \>  \>  $s$:= first such school on  $a_i$'s list ;\\
\>  \>  \>  $\MM:=\MM\cup \{(a_i,s)\}$; \\
\>  \> \}\\
{\bf\ end}
\end{tabbing}
\vspace{-1ex}
\noindent\rule{14cm}{0.5pt}

\caption{Algorithm Serial Dictatorship}
\label{alg1}
\end{figure}
}
\begin{thm}
Let $J$ be an instance of  {\sc tap-am} with the master list $a_1,a_2, \dots,a_n$ of applicants. Then $J$ admits a unique stable matching that may be found by an application of Algorithm Serial Dictatorship as shown in Figure \ref{alg1}.
\end{thm}
\begin{proof} It is easy to see that Serial Dictatorship outputs a  matching.
We have to prove that this matching is stable and that it is the unique stable matching.
  
{\bf $\MM$ is stable. } Suppose that $(a_i,s_j)$ is a blocking pair and that $i$ is the smallest index of an applicant involved in a blocking pair. Since $a_i$ has chosen the best available school from her list, $s_j$ did not have enough capacity to accept $a_i$ when it was $a_i$'s turn in the algorithm. However, all the applicants that were assigned to $s_j$ at that moment precede $a_i$ in the master list, hence   $(a_i,s_j)$ cannot be a blocking pair.

{\bf Uniqueness.} Let $\MM'\ne \MM$ be another stable matching and let $a_i$ be the first applicant in the master list with $\MM'(a_i)\ne \MM(a_i)$. As Serial Dictatorship gave $a_i$ her best available school and all applicants who precede $a_i$ in the master list have the same assignments in $\MM$ as in $\MM'$, it must be the case that $a_i$ prefers $s_j=\MM(a_i)$ to $s_k=\MM'(a_i)$. But this implies that $(a_i,s_j)$ is a blocking pair for $\MM'$, as $s_j$ will be able to reject one or two applicants worse than $a_i$ in order to free up sufficient capacity for $a_i$ (for, $s_j$ had enough room for $a_i$ in $\MM$ when it was $a_i$'s turn during Serial Dictatorship, and any applicant that  precedes  $a_i$ in the master list has the same assignment in $\MM$ as she does in $\MM'$). 
\end{proof}

The situation with a master list of schools, although less likely to occur in practice, is also efficiently solvable.

{\small
\begin{figure}
\noindent\rule{14cm}{0.5pt}
\vspace{-2ex}
\begin{tabbing}
~~~ \= ~~~ \= ~~~ \= ~~~ \= \kill
{\bf begin}\\
\> $\MM:=\emptyset$;\\
\> {\bf for\ } $j=1,2,\dots,m$ \\
\> \> /* let $s_j$'s list be $a_{i_1},\dots, a_{i_\ell}$ */ \\
\> \> {\bf for\ } $r=1,2,\dots,\ell$ \\
\> \> \> {\bf if} $a_{i_r}$ is unassigned and $s_j$ has enough capacity for $a_{i_r}$ {\bf then}\\
\> \> \> \> $\MM:=\MM\cup \{(a_{i_r},s_j)\}$; \\
{\bf end}
\end{tabbing}
\vspace{-2ex}
\noindent\rule{14cm}{0.5pt}
\vspace{-2ex}
\caption{Algorithm Dual Serial Dictatorship}
\label{alg2}
\end{figure}
}
\begin{thm}
Let $J$ be an instance of  {\sc tap-sm} with the master list of schools $s_1,s_2, \dots,s_m$. Then $J$ admits a unique stable matching that may be found  by an application of Algorithm Dual Serial Dictatorship as shown in Figure \ref{alg2}.
\end{thm}
\begin{proof} 
Let us denote by $J(s_1)$ the subinstance of $J$ containing just school $s_1$ and applicants who apply to $s_1$. $J(s_1)$ is an instance of {\sc tap-am}, so it has a unique stable matching. This is obtained by Serial dictatorship of applicants that is equivalent with the part of Dual serial dictatorship within one iteration of the {\bf for}-cycle for schools. Let us denote this matching by $\MM_1$. Let us further observe that no applicant assigned to $s_1$ could be a member of a blocking pair, as she received her most preferred school. If we now denote by $J(-s_1)$ the subinstance of $J$ with pairs of $\MM_1$ deleted, the result follows by induction.
\end{proof}

\section{Subject-specific preference lists}\label{s_subject_specific}

Suppose that a school can order applicants differently for different subjects. The  definition of a blocking pair should now be modified in order to take account of this scenario.

\begin{df}\label{d_blok2} Let $\MM$ be a matching. We say  that a pair $(a,s)$ with $\pp(a)=\{p_1,p_2\}$ and  $s\in S(A)$  is blocking if
$a$ is not assigned in $\MM$ or $a$ prefers $s$ to $\MM(a)$, and one of the following conditions hold:
\begin{itemize}\itemsep0pt
\item[$(i)$] $s$ is undersubscribed in both $p_1$ and $p_2$, 
\item[$(ii)$] $s$ is undersubscribed in $p_i$ and it prefers $a$ in subject $p_{3-i}$  to one applicant in $\MM_{p_{3-i}}(s)$ for some $i\in\{1,2\}$, 
\item[$(iii)$] $s$ prefers $a$ in both subjects $p_1,p_2$ to one applicant in $\MM_{p_1,p_2}(s)$,
\item[$(iv)$] $s$ prefers $a$ in subject $p_1$ to applicant $a_1\in \MM_{p_1}(s)$ and in subject $p_2$ to  another applicant $a_2\in \MM_{p_2}(s)$.
\end{itemize}
\end{df}

Let us denote this variant by {\sc tap-ss}.  Define a matching $\MM$ in a {\sc tap-ss} instance to be \emph{applicant-complete} if every applicant is assigned in $\MM$.

\begin{lema}
Given an instance of {\sc tap-ss}, the problem of deciding whether an applicant-complete stable matching exists is NP-complete. This result holds even if there are at most three subjects, each partial capacity of a school is at most 1, and the preference lists of the schools are derived from subject-specific master lists of the applicants.
\label{lem1}
\end{lema}
\begin{proof} We reduce from \sat (see the proof of Theorem \ref{thm1}).
Let $B$ be an instance of this problem, where $V = \{v_0, v_1,\dots, v_{n-1}\}$ and $C =\{c_1, c_2,\dots, c_m\}$ be the set of variables and clauses respectively in $B$. We construct an instance $J$ of {\sc tap} in the following way.  

There are 3 subjects, namely $F$, $I$ and $M$. For each variable $v_i$ ($0\leq i\leq n-1$) there are 4 applicants $x_{4i+r}$ ($0\leq r\leq 3$), each of type $FI$, and 4 schools $y_{4i+r}$ ($0\leq r\leq 3$).  For each clause $c_j$ ($1\leq j\leq m$) there are 4 applicants  $q_j$ and $w_j^t$ ($1\leq t\leq 3$), each of type $FM$, and 4 schools $s_j^t$ ($1\leq t\leq 4$).  Let $X=\{x_i : 0\leq i\leq 4n-1\}$, $Y=\{y_i : 0\leq i\leq 4n-1\}$, $W=\{w_j^t : 1\leq j\leq m\wedge 1\leq t\leq 3\}$, $Q=\{q_j : 1\leq j\leq m\}$, $S'=\{s_j^t : 1\leq j\leq m\wedge 1\leq t\leq 3\}$, $S''=\{s_j^4 : 1\leq j\leq m\}$ and $S=S'\cup S''$. 

For each $i$ ($0\leq i\leq n-1$), applicants $x_{4i}$ and $x_{4i+1}$ correspond to the first and second occurrences of literal $v_i$ in $B$, and applicants $x_{4i+2}$ and $x_{4i+3}$ correspond to the first and second occurrences of literal $\bar{v_i}$ in $B$, respectively.  For each $r\in \{0,1\}$, let $s(x_{4i+r})$ denote the school $s_j^t$ such that the $(r+1)$th occurrence of literal $v_i$ appears in position $t$ of clause $c_j$ ($1\leq j\leq m$, $1\leq t\leq 3$). Similarly, for each $r\in \{2,3\}$, let $s(x_{4i+r})$ denote the school $s_j^t$ such that the $(r-1)$th occurrence of literal $\bar{v_i}$ appears in position $t$ of clause $c_j$ ($1\leq j\leq m$, $1\leq t\leq 3$).

\begin{figure}[ht]
\begin{displaymath}
\begin{array}{rllrl}
\mbox{applicant} &\mbox{type} &\mbox{preferences}  & \mbox{school} & \mbox{capacities} \\
&&&& F\quad I\quad M \\ 
\noalign{\hrule} 
\\[-8pt]
x_{4i}   & FI & y_{4i},s(x_{4i}),y_{4i+1}                   & y_{4i} & 1\quad\ 1\quad\ 0 \\[1.5mm]
x_{4i+1} & FI & y_{4i+1},s(x_{4i+1}),y_{4i+2} \qquad \qquad & y_{4i+1} & 1\quad\ 1\quad\ 0 \\[1.5mm]
x_{4i+2} & FI & y_{4i+3},s(x_{4i+2}),y_{4i+2}               & y_{4i+2} & 1\quad\ 1\quad\ 0 \\[1.5mm]
x_{4i+3} & FI & y_{4i},s(x_{4i+3}),y_{4i+3}                 & y_{4i+3} & 1\quad\ 1\quad\ 0 \\[1.5mm]
q_j      & FM & s_j^1,s_j^2,s_j^3                           & s_j^t    & 1\quad\ 1\quad\ 1 \\[1.5mm]
w_j^t    & FM & s_j^t,s_j^4                                 & s_j^4    & 1\quad\ 0\quad\ 1 \\[1.5mm]
\noalign{\hrule}
\end{array}
\end{displaymath}
\caption{The {\sc tap-ss} instance constructed in the proof of Lemma \ref{lem1}}\label{t_2}
\end{figure} 

The applicants' preferences, together with a summary of their types and a summary of the schools' partial  capacities, are given in Figure \ref{t_2}.  Here, the subscripts and superscripts involving $i$, $j$ and $t$ range over the following intervals: $0\leq i\leq n-1$, $1\leq j\leq m$ and $1\leq t\leq 3$.

\begin{figure}[ht]
\begin{displaymath}
\begin{array}{rl}\\
F: & \langle W\rangle ~~ \langle X\rangle ~~ \langle Q\rangle \\
I: & \langle \bar{X}\rangle \\
M: & \langle Q\rangle ~~ \langle \bar{W}\rangle
\end{array}
\end{displaymath}
\caption{The master lists for the {\sc tap-ss} instance constructed in the proof of Lemma \ref{lem1}}\label{t_3}
\end{figure}

We now construct the subject-specific master lists of applicants.  Let $\langle X\rangle$ denote the elements of $X$ in increasing order of subscript, and let $\langle \bar{X}\rangle$ denote the reverse of this order.  Similarly let $\langle W\rangle$ denote the elements of $W$ listed in increasing order of subscript, and within this ordering, those elements with equal subscript are listed in increasing order of superscript.  Also let $\langle \bar{W}\rangle$ denote the reverse of $\langle W\rangle$.  Finally let $\langle Q\rangle$ denote the elements of $Q$ listed in increasing order of subscript.  The master lists of the applicants with respect to subjects are shown in Figure \ref{t_3}.

For each $i$ ($0\leq i\leq n-1$), let us denote $T_i=\{(x_{4i+r},y_{4i+r}) : 0\leq r\leq 3\}$ and $F_i=\{(x_{4i+r},y_{4i+r+1}) : 0\leq r\leq 2\}\cup \{(x_{4i+3},y_{4i})\}$.  We claim that $B$ has a satisfying truth assignment if and only if $J$ has an applicant-complete stable matching.

For, let $f$ be a satisfying truth assignment of $B$. Define a matching $\MM$ in $J$ as follows.
For each variable $v_i\in V$, if $f(v_i)={\sf true}$, add the pairs in $T_i$ to $\MM$ and if $f(v_i)={\sf false}$, add the pairs in $F_i$ to $\MM$.  Each clause $c_j\in C$ contains some literal in $c_j$ that is true under $f$,  let $t$ be the position of $c_j$ containing this true literal ($1\leq t\leq 3$).  Add the following pairs to $\MM$: 
$$\{(w_j^{t'},s_j^{t'}) : 1\leq t'\leq 3\wedge t\neq t'\}\cup \{(q_j,s_j^t),(w_j^t,s_j^4\}.$$ 

It is obvious that the defined assignment is an applicant-complete matching; it remains to prove that it is stable.  It is straightforward to verify that no applicant in $Q$ can be involved in a blocking pair of $\MM$, and no pair in $X\times Y$ can block $\MM$.  Now suppose that $(w_j^t,s_j^t)$ blocks $\MM$ for some $j$ ($1\leq j\leq m$) and $t$ ($1\leq t\leq 3$).  Then $(w_j^t,s_j^4)\in \MM$ and $(q_j,s_j^t)\in \MM$, and school $s_j^t$ prefers $q_j$ over $w_j^t$ for subject $M$, so $(w_j^t,s_j^t)$ cannot block $\MM$ after all.  Finally suppose that $(x_{4i+r},s(x_{4i+r}))$ blocks $\MM$ for some $i$ ($0\leq i\leq n-1$) and $r$ ($0\leq r\leq 1$).  Then $f(v_i)={\sf false}$ by construction of $\MM$.  Let $s_j^t=s(x_{4i+r})$.  Then $(q_j,s_j^t)\in \MM$, since $(x_{4i+r},s(x_{4i+r}))$ blocks $\MM$.  But then by construction of $\MM$, the $t$th literal of $c_j$ is true under $f$, a contradiction.  The argument is similar if $r\in \{2,3\}$. Hence $\MM$ is stable in $J$.

Conversely suppose that $\MM$ is an applicant-complete stable matching in $J$.  
For any $j$ ($1\leq j\leq m$), it follows that $(q_j,s_j^t)\in \MM$ for some $t$ ($1\leq t\leq 3$) and thus $(w_j^t,s_j^4)\in \MM$, since $q_j$ and $w_j^t$ must be assigned.  Moreover $(w_j^{t'},s_j^{t'})\in \MM$ for each $t'$ ($1\leq t'\leq 3$, $t\neq t'$).  Thus $\MM$ contains no pair of the form $(x_{4i+r},s(x_{4i+r}))$ ($0\leq i\leq n-1$, $0\leq r\leq 3$).  Moreover, since each member of $X$ must be assigned in $\MM$, we have thus established that for each $i$ ($0\leq i\leq n-1$), either $T_i\subseteq \MM$ or $F_i\subseteq \MM$.  

Now we construct a truth assignment $f$ in $B$ as follows.  
If $T_i\subseteq \MM$ set  $f(v_i)={\sf true}$ and if  $F_i\subseteq \MM$ set $f(v_i)={\sf false}$.  
We claim that $f$ is a satisfying truth assignment.  For, suppose that some clause $c_j$ contains no true literal.  As $\MM$ is an applicant-complete matching, $(q_j,s_j^t)\in \MM$ for some $t$ ($1\leq t\leq 3$).  Now let $x_{4i+r}$ be the applicant such that $s(x_{4i+r})=s_j^t$ ($0\leq i\leq n-1$, $0\leq r\leq 3$).  If $r\in \{0,1\}$ then $f(v_i)={\sf false}$, so $F_i\subseteq \MM$.  Hence $(x_{4i+r},s(x_{4i+r}))$ blocks $\MM$, a contradiction.  Similarly if $r\in \{2,3\}$ then $f(v_i)={\sf true}$, so $T_i\subseteq \MM$.  Hence $(x_{4i+r},s(x_{4i+r}))$ blocks $\MM$, again a contradiction.
\end{proof}

\begin{lema}
Given an instance of {\sc tap-ss}, the problem of deciding whether a stable matching exists is NP-complete. This result holds even if there are at most three subjects, each partial capacity of a school is at most 1, and the preference lists of the schools are derived from subject-specific master lists of the applicants.
\label{lem2}
\end{lema}
\begin{proof}
We show how to modify the reduction presented in the proof of Lemma \ref{lem1} in order to ensure that any stable matching in $J$ is applicant-complete.  We create a new {\sc tap} instance $J'$ from $J$ as follows.  For each applicant $a$ in $J$, create two new applicants $a'$ and $a''$.  If $a$ is of type $FI$, then $a'$ and $a''$ are of type $FM$ and $IM$ respectively.  If $a$ is of type $FM$, then $a'$ and $a''$ are of type $FI$ and $IM$ respectively.  Create a new school $g(a)$ which has capacity 1 for each of subjects $F$, $I$ and $M$.  Append $g(a)$ to applicant $a$'s preference list in $J$ to obtain her preference list in $J'$.  Each of applicants $a'$ and $a''$ finds only $g(a)$ acceptable.

Let $X'$ and $X''$ denote the sets of newly-created applicants in $J'$ with single and double primes respectively that correspond to applicants in $X$.  Define $R'$ and $R''$ similarly for the newly-created applicants in $J'$ that correspond to applicants in $Q\cup W$.  For $A\in \{R,X\}$, let $\langle A'\rangle$ and $\langle A''\rangle$ denote arbitrary but fixed orderings of the applicants in $A'$ and $A''$ respectively.

The subject-specific master lists in $J'$ are as shown in Figure \ref{t_4}.
\begin{figure}[ht]
\begin{displaymath}
\begin{array}{rl}
F: & \langle W\rangle ~~ \langle X\rangle ~~ \langle Q\rangle ~~ \langle X'\rangle ~~ \langle R'\rangle\\
I: & \langle X''\rangle  ~~ \langle R'\rangle ~~ \langle R''\rangle ~~ \langle \bar{X}\rangle \\
M: & \langle R''\rangle ~~ \langle Q\rangle ~~ \langle \bar{W}\rangle ~~ \langle X'\rangle ~~ \langle X''\rangle
\end{array}
\end{displaymath}
\caption{The master lists of the subjects in the {\sc tap-ss} instance constructed in the proof of Lemma \ref{lem2}}\label{t_4}
\end{figure}

We show how to modify the proof of Lemma \ref{lem1} to show that $B$ has a satisfying truth assignment if and only if $J'$ has a stable matching.

Firstly, if $f$ is a satisfying truth assignment of $B$, construct the matching $\MM$ in $J$ as in the proof of Lemma \ref{lem1}.  We then extend $\MM$ to a matching $\MM'$ in $J'$ as follows.  For each applicant $a$ in $J$, add the pair $(a',g(a))$ to $\MM$.  Since $\MM$ is applicant-complete in $J$, it is straightforward to verify that $\MM'$ is stable in $J'$.

Conversely suppose that $\MM'$ is a stable matching in $J'$.  We firstly claim that each applicant $a$ in $J$ is assigned in $\MM'$ to a school better than $g(a)$.  For, suppose $(a,g(a))\in \MM'$.  Then $(a'',g(a))$ blocks $\MM'$, a contradiction.  Now suppose that $a$ is unassigned in $\MM'$. Clearly some applicant is assigned to $g(a)$ in $\MM'$, for otherwise $(a,g(a))$ blocks $\MM'$.  If $(a',g(a))\in \MM'$ then $(a,g(a))$ blocks $\MM'$, whilst if $(a'',g(a))\in \MM'$ then $(a',g(a))$ blocks $\MM'$.  The claim is thus proved.  It is then also straightforward to verify that $(a',g(a))\in \MM'$ for each applicant $a$ in $J$, for otherwise $(a',g(a))$ blocks $\MM'$.

Let $\MM$ be the matching obtained from $\MM'$ by removing all pairs of the form $(a',g(a))$, where $a$ is an applicant in $J$.  It follows by the previous paragraph that $\MM$ is an applicant-complete stable matching in $J$.  The remainder of the proof is then identical to the converse direction of the proof of Lemma \ref{lem1}.
\end{proof}
\begin{thm}
\label{tap-ss-npcomplete}
Given an instance of {\sc tap-ss}, the problem of deciding whether a stable matching exists is NP-complete. This result holds even if there are at most three subjects, each partial capacity of a school is at most 1, the preference lists of the schools are derived from subject-specific master lists of the applicants, and the preference lists of the applicants are derived from a single master list of schools.
\end{thm}
\begin{proof}
We consider the reduction given by Lemmata \ref{lem1} and \ref{lem2}, and show that the applicants' preference lists may be derived from a single master list of schools.

For each $i$ ($0\leq i\leq n-1$), let $\langle S_i\rangle$ denote the sequence
$$\langle y_{4i},s(x_{4i}),s(x_{4i+3}),y_{4i+3},s(x_{4i+2}),y_{4i+1},s(x_{4i+1}),y_{4i+2}\rangle.$$
Let $S^4=\{s_j^4 : 1\leq j\leq m\}$ and let $\langle S^4\rangle$ denote an arbitrary order of the schools in $S^4$.  Let $G$ denote the set of schools of the form $g(a)$ as introduced in the proof of Lemma \ref{lem2} for each applicant $a$ in the original {\sc tap-ss} instance as constructed in the proof of Lemma \ref{lem1}.  Let $\langle G\rangle$ denote an arbitrary order of the schools in $G$.  Define the following master list of schools:
$$\langle S_0\rangle ~~ \langle S_1\rangle ~~ \dots ~~ \langle S_{n-1}\rangle ~~ \langle S^4\rangle ~~ \langle G\rangle$$

In the proof of Lemma \ref{lem1}, let the preference list of each $q_j$ ($1\leq j\leq m$) be reordered such that the relative ordering of the three schools $s_j^1$, $s_j^2$ and $s_j^3$ is derived from the above master list.  This does not change the remainder of the proof of Lemma \ref{lem1}, nor the proof of Lemma \ref{lem2}.  Moreover every other applicant's preference list is derived from the above master list of schools.  The theorem then follows.
\end{proof}

\section{``Most stable'' matchings}
\label{sec:moststable}
Given an instance of {\sc tap}, we have already seen that a stable matching need not exist.  In such cases it is natural to seek a matching that is ``as stable as possible'' in a precise sense.  Here we regard such ``most stable'' matchings as those that admit the minimum number of blocking pairs.  Note that this approach was also considered in \cite{ABM06}.

Clearly {\sc min bp tap}, the problem of finding such a matching in an instance of {\sc tap}, is NP-hard, given Theorem \ref{thm1}.  It is then natural to consider the approximability of this problem.  In this section we prove that the problem of finding a matching with the minimum number of blocking pairs, given an instance of {\sc tap}, is NP-hard and not approximable within $n^{1-\varepsilon}$, for any $\varepsilon>0$, unless P=NP, where $n$ is the number of agents (applicants plus schools).  We show that similar observations hold for {\sc min bp tap-ss}, the problem of finding a matching with the minimum number of blocking pairs in an instance of {\sc tap-ss}.

We begin by stating a more general result.  Let {\sc p} be any stable matching problem for which deciding whether a stable matching exists is NP-complete. Define {\sc min bp p} to be the problem of finding a matching with the minimum number of blocking pairs, given an instance $I$ of {\sc p}.  We let $opt(I)$ denote 1 plus the minimum number of blocking pairs admitted by any matching in $I$.  (We then use $opt(I)$ as the measure against which any approximation algorithm for {\sc p} performs, since adding 1 ensures that $opt(I)\geq 1$.)  We now present a general inapproximability result for {\sc min bp p}, and later use it to derive similar inapproximability results for {\sc min bp tap} and {\sc min bp tap-ss}.

\begin{thm}
\label{GENERAL-APPROX-NPCOMPLETE}
{\sc min bp p} is not approximable within $n^{1- \varepsilon}$, where $n$ is the number of agents in a given instance, for any $\varepsilon > 0$ unless P=NP.
\end{thm}
\begin{proof}
Let $n_0$ be the number of agents in an instance $I$ of {\sc p}. Choose $c = \lceil 2 / \varepsilon \rceil$ and $k= n_0^c$. Now, let $I_1 , I_2 , \ldots , I_k$ be $k$ disjoint copies of the instance $I$. Let $I^{\prime}$ be the instance of {\sc p} formed by taking the union of the sub-instances $I_1 , I_2 , \ldots , I_k$. Let $n = kn_0 $ denote the number of agents in $I^{\prime}$.

Clearly if $I$ admits a stable matching then each copy of $I$ must admit a stable matching. Hence $opt(I^{\prime }) = 1$. However, if $I$ admits no stable matching, then each $I_r ~ (1\leq r\leq k)$ admits only matchings with one or more blocking pair, and hence any matching admitted by $I^{\prime}$ must admit $k$ or more blocking pairs. Hence $opt(I^{\prime }) \geq k+1$. We now show that $n^{1- \varepsilon} \leq k$.

Firstly $n = k n_0 \leq 2k n_0 = 2 n_0^{c+1}$. Hence

$$\dfrac{n} {2} \leq n_0^{c+1}$$
which implies 
$$\left (\dfrac{n}{2}\right )  ^{1/(c+1)} \leq n_0.$$
Since $k = n_0^c$ it follows that 
$$\left (\dfrac{n}{2}\right)^{c/(c+1)} \leq k$$
and hence 

\begin{equation} \label{GeneralInapp-Eq1} \displaystyle 2^{-c/(c+1)}n^{c/(c+1)} \leq k. \end{equation}

We know that $n = k n_0 = n_0^{c+1} \geq n_0^c$ and we lose no generality by assuming that $n_0 \geq 2$. Hence $n \geq 2^c$ and it follows that $n^{-1} \leq 2^{-c}$
and thus 
\begin{equation} \label{GeneralInapp-Eq2} \displaystyle n^{-1/(c+1)} \leq 2^{-c/(c+1)}. \end{equation}
Hence 
\begin{equation} \label{GeneralInapp-Eq3} \displaystyle n^{-1/(c+1)} n^{c/c+1} \leq 2^{-c/(c+1)}n^{c/(c+1)} \end{equation}
and it follows that 
\begin{equation} \label{GeneralInapp-Eq4} n^{(c-1)/(c+1)} = n^{c/(c+1)} n^{-1/(c+1)} \leq 2^{-c/(c+1)} n^{c/(c+1)} \leq k. \end{equation}

\noindent We now show that $n^{1- \varepsilon} \leq n^{(c-1)/(c+1)}$. Observe that $c \geq 2 / \varepsilon$ and thus $c + 1 \geq 2 / \varepsilon$. Hence 
$$1 - \varepsilon \leq 1 - \dfrac{2}{c+1}\leq  \dfrac{c-1}{c+1}$$
and hence by Inequality \ref{GeneralInapp-Eq4}, $n^{1- \varepsilon} \leq k$.

Now, assume that $X$ is an approximation algorithm for {\sc p} with a performance guarantee of $n^{1 - \varepsilon}$. If $I$ admits a stable matching, $X$ must return a matching in $I^{\prime }$ with measure at most $opt(I).n^{1 - \varepsilon}=n^{1 - \varepsilon}\leq k$, since $opt(I)=1$. Otherwise, $I^{\prime}$ does not admit a stable matching and, as shown above, $X$ must return a matching with measure $\geq k + 1$. Thus algorithm $X$ may be used to determine whether $I$ admits a stable matching in polynomial time, a contradiction. Hence, no such polynomial approximation algorithm can exist unless P = NP.
\end{proof}

Theorem \ref{GENERAL-APPROX-NPCOMPLETE} allows us to state the following corollaries for the inapproximability of {\sc min bp tap} and {\sc min bp tap-ss}. 

\begin{cor}
\label{minbptap_innapproximabilityresult1}
{\sc min bp tap} is not approximable within $n^{1- \varepsilon}$, where $n$ is the number of agents in a given instance, for any $\varepsilon > 0$ unless P=NP. This result holds even if there are at most three subjects, each partial capacity of a school is at most 2, and the preference list of each applicant is of length at most 3.
\end{cor}
\begin{proof}
This result follows immediately from Theorems \ref{thm1} and \ref{GENERAL-APPROX-NPCOMPLETE}.
\end{proof}

\begin{cor}
{\sc min bp tap-ss} is not approximable within $n^{1- \varepsilon}$, where $n$ is the number of agents in a given instance, for any $\varepsilon > 0$ unless P=NP. This result holds even if there are at most three subjects, each partial capacity of a school is at most 1, the preference lists of the schools are derived from subject-specific master lists of the applicants, and the preference lists of the applicants are derived from a single master list of schools.
\end{cor}
\begin{proof}
This result follows immediately from Theorems \ref{tap-ss-npcomplete} and \ref{GENERAL-APPROX-NPCOMPLETE}.
\end{proof}

\section{Conclusion}\label{s_conclusion}
In this paper we studied the notion of stability in the Teachers Assignment Problem from the computational perspective. Notice that when there are only two different subjects in total, then all the applicants are of the same type and the problem reduces to the classical Hospitals / Residents problem. However, the presence of just three subjects makes the problem of deciding the existence of a stable matching intractable, unless there is a master list of applicants or  a master list of schools. By contrast, if the schools may order applicants differently in each of their specialization subjects, then even subject-specific master lists plus the master list of schools do not make the problem tractable. Thus we have been able to characterize the borderline between the polynomially solvable and intractable cases of the Teachers Assignment Problem.   

Further, we proved a general inapproximability result for the problem of finding a matching with the minimum number of blocking pairs. This implies that, if $P\ne NP$, there is no polynomial approximation algorithm with performance guarantee $n^{1-\varepsilon}$ for finding a matching with the minimum number of blocking pairs in both preference variants of the Teachers Assignment Problem.

\end{document}